\documentclass[3p,twocolumn]{elsarticle}

\journal{Computers \& Chemical Engineering}
\usepackage{graphicx} 
\usepackage[dvips]{epsfig}
\usepackage{amssymb}
\usepackage{amsfonts}
\usepackage{amsmath}
\usepackage{amsthm}
\usepackage[english]{babel}
\usepackage[utf8]{inputenc}
\usepackage{xcolor}
\usepackage{footnote}
\usepackage{subfig}

\usepackage{siunitx}
\usepackage{eurosym}
\usepackage{tikz}

\newcommand {\bmat} {\left[\begin{array} }
\newcommand {\emat} {\end{array}\right]}
\newcommand {\ematrix}{\end{array}\right]}
\newcommand{\blista}{\renewcommand{\labelenumi}{(\roman{enumi})}
\begin{enumerate}}
\newcommand{\elista}{\end{enumerate} \renewcommand{\labelenumi}{\arabic{enumi}.}}
\newcommand {\beqn}{\begin{equation}}
\newcommand {\eeqn}{\end{equation}}
\newcommand {\beqna}{\begin{eqnarray}}
\newcommand {\eeqna}{\end{eqnarray}}
\newcommand {\beqnan}{\begin{eqnarray*}}
\newcommand {\eeqnan}{\end{eqnarray*}}
\newcommand {\nn}{\nonumber}
\newcommand{\vx} {\mathbf{x}}
\newcommand{\vu} {\mathbf{u}}

\newcommand{\vz} {\mathbf{z}}
\newcommand{\vy} {\mathbf{y}}

\newcommand{\setX}{\mathcal{X}}
\newcommand{\setY}{\mathcal{Y}}
\newcommand{\setU}{\mathcal{U}}
\newcommand{\setZ}{\mathcal{Z}}

\newcommand{\setD}{\mathcal{D}}
\newcommand{\Or}{\mathcal{O}}

\def \R{ {\mathbb{R}} }
\def \I{ {\mathbb{I}} }
\newtheorem{assumption}{Assumption}
\newtheorem{theorem}{Theorem}
\newtheorem{corollary}{Corollary}
\newtheorem{lemma}{Lemma}
\newtheorem{remark}{Remark}
\newtheorem{proposition}{Proposition}

\newtheorem{definition}{Definition}

\begin{document}

\begin{frontmatter}

\title{Oracle-based economic predictive control}

\author{Jos\'e Mar\'ia Manzano\fnref{ca}}
\fntext[ca]{~Corresponding author}
\ead{jmanzano@uloyola.es}
\author{David Mu\~noz de la Pe\~na$^2$}
\ead{dmunoz@us.es}
\author{Daniel Limon$^2$}
\ead{dlm@us.es}
\address{$^1$ Departamento de Ingenier\'ia, Universidad Loyola Andaluc\'ia \\Avda. de las Universidades s/n 41071 Dos Hermanas, Sevilla, Spain}
\address{$^2$Departamento de Ingenier\'ia de Sistemas y Autom\'atica, Universidad de Sevilla\\ Camino de los Descubrimientos s/n 41092 Sevilla, Spain}

\begin{abstract}
This paper presents an economic model predictive controller, under the assumption that the only measurable signal of the plant is the economic cost to be minimized. In order to forecast the evolution of this economic cost for a given input trajectory, a prediction model with a NARX structure, the so-called \emph{oracle}, is proposed. Sufficient conditions to ensure the existence of such oracle are studied,  proving that it can be derived for a general nonlinear system if the economic cost function is a Morse function. Based on this oracle, economic model predictive controllers are proposed, and their stability is demonstrated in nominal conditions under a standard dissipativity assumption. The viability of these controllers in practical settings (where the oracle may provide imperfect predictions for generic inputs) is proven by means of input-to-state stability. These properties have been illustrated in a case study based on a continuously stirred tank reactor.
%
\end{abstract}


\end{frontmatter}


\section{Introduction}

Often, control systems have to simultaneously consider both performance and safety requirements. This double objective has been typically addressed by means of a hierarchical structure, where a real time optimization layer calculates the equilibrium point that minimizes the operation cost, while a lower control layer regulates the system to this equilibrium point. Recently, in the model predictive control (MPC) framework, this hierarchical control structure has been united in a single layer~\cite{dorfler2015breaking}, aimed to minimize the operation cost during the transient, instead of a tracking cost, often designed to provide robustness and stability properties. This is the so-called economic MPC, whose properties have been studied in several works \cite{RawlingsCDC12, ChristofidesEMPC_AIChe2012}. The main difference between economic MPC  and regulation MPC is that the former relaxes the architecture of the optimization problem, in order to minimize any \emph{economic} cost function, which may not be positive definite. Two recent overviews on EMPC can be found in~\cite{ellis2017economic,Faulwasser2018economic}.  Further, extensions to robust and multistage or output-feedback EMPC can be found in~\cite{subramanian2015economic,BayerJPC14, lucia2014handling}, among others.

Predictive controllers are based on the availability of  a model of the plant, in order to predict the evolution of the states of the system, and based on these predictions, the cost to be minimized is calculated. Recently, data-based MPCs have been proposed to address applications in which an appropriate model of the plant is not available, and hence the predictions have to be obtained from historical data sets, as reviewed in~\cite{hewing2020learning}. Accordingly, this data-based approach is also being applied to economic MPCs. For example, in~\citep{KheradmandiMatMDPI18}, a Lyapunov-based economic MPC (LB-EMPC) that integrates a linear prediction model is presented, updated online from the measurements of the plant.  In~\citep{GiulianiSSMS18}, a  LB-EMPC for nonlinear systems aimed to maintain excitation on the system in order to obtain an state-space model from the measured inputs was presented. In~\citep{WuMatMDPI19}, the authors propose a recurrent neural network to learn the model of the plant that is controlled by an LB-EMPC, and they extend this approach to take into account constraints by means of barrier functions in~\citep{WuChERD19}. In~\cite{GrosTAC19}, a data-driven EMPC based on reinforcement learning in which the EMPC is used as approximator of the value function of the reinforced learning policy is proposed.

These data-based approaches are based on both inputs and outputs historical data sets. However, there may exist situations in which no measurements of any inner variable of the plant are available, for example, to maintain privacy of operation, or due to security reasons. Consider for instance a data center in which the operation cost accounts for the cost of the electric consumption of the refrigeration system and the consumption of the servers. In order to design a controller to optimize the operation cost, sharing inner information of the state of the servers could be limited due to security reasons, while sharing only the operation cost may not jeopardize the security of the system.

In this paper, we study the case in which the only available measurement from the plant is the value of the economic cost to be optimized. The prediction of the behaviour of the plant is carried out by an oracle that forecasts the economic performance of the plant, from a historical data set of the tuple inputs-economic costs, using a nonlinear autoregressive exogenous model (NARX) structure, which is obtained using nonlinear identification or machine learning methods~\citep{ChiusoARC19}. The existence of this class of oracles for the prediction of the economic cost is studied, proving that they can be derived for general nonlinear plants under mild assumptions on the economic cost function. Based on this oracle, several economic model predictive controllers are proposed, proving that in the case of exact predictions, they inherit the properties of the economic MPC based on the process model. In addition, in the practical case of an imperfect oracle, it is proven that the oracle-based economic MPC with terminal cost function is input-to-state stable with respect to the estimation error of the oracle.

To the best of the authors' knowledge, this is the first work in which an economic MPC based only on the measure of the economic cost is studied, proving that for a general nonlinear system and under mild assumptions on the economic cost function, it suffices to measure the economic cost function to design a stabilizing economic predictive controller using a data-based oracle. The proposed controller has been applied in simulation to the economic operation of a continuously stirred tank reactor, using kinky inference processes~\cite{manzano2019output} to learn the oracle. 
A preliminary version of this work was presented in~\cite{manzano2019oracle}.

The rest of the paper is structured as follows: Section~\ref{sec:problem} presents the problem formulation and the standard economic MPC. Section~\ref{sec:oracle} states the conditions under which it is possible to define an oracle to predict the future evolution of the economic cost. Section~\ref{sec:oempc} describes the proposed oracle-based economic predictive controllers and Section~\ref{sec:practical} addresses practical aspects of the problem. Finally, Section~\ref{sec:casestudy} presents the case study.

\subsection*{Notation}
Given two column vectors $v$ and $w$, $(v,w)$ stands for $[v^T,w^T]^T$.
The set~$\mathbb I_a^b$ stands for the set of integers from $a$ to $b$.
A function $\alpha: \mathbb R_{\geq 0} \rightarrow \mathbb R_{\geq 0}$ is a~$\mathcal K$-function if it is strictly increasing and $\alpha(0)=0$. Besides, if a $\mathcal K$-function is such that $\lim\limits_{s \rightarrow \infty}\alpha(s)=\infty$, then it is called a $\mathcal K_\infty$-function. Given a set $\setX \subseteq \R^n$, 
$\setX^M$ denotes the cartesian product of the set $M$ times, i.e. $\setX^{M}=\setX^{M-1} \times \setX$ with $\setX^1=\setX$.


\section{Problem formulation}\label{sec:problem}

In this paper, we consider that the system to be controlled is a sampled continuous-time system described by an unknown discrete time model
\beqna
x(k+1)&=&f(x(k),u(k)),
\label{EqMdl}
\eeqna
where $x(k) \in \R^n$ is the state of the plant and $u(k) \in \R^m$ is the control input. It is assumed that the inputs are subject to (hard) constraints~\mbox{$u(k) \in \setU$}, where~\mbox{$\setU \subset \R^m$} is a compact set.

The objective of the control strategy to be designed is to guarantee that the closed-loop system is stable, while a cost function is minimized during the transient. This cost function is said to be economic because it measures the performance of the evolution of the system according to a  generic function, that does not necessarily penalize the tracking error w.r.t. a given target. 

The economic cost function to be considered in this paper is defined by a function of the form~$\ell(x,u)$.
The model~$f$ and the economic cost function~$\ell$ must satisfy the following condition:

\begin{assumption} \label{HipCont}
	The function $f(x,u)$ is smooth and  state invertible, i.e., for a given $u$, $f$ defines a diffeomorphism in $x$. The function $\ell(\cdot,\cdot): \R^n \times \setU \mapsto \mathcal L$ is smooth and its image $\mathcal L\subset\R$ is a compact set.
\end{assumption}

Note that as stated in~\cite{JakubczykSJCO90}, in the general case in which the model function~\eqref{EqMdl} is derived from sampling a continuous-time system controlled using a zero-order holder, such that it is described by a finite-dimensional differential equation with an unique solution, the resulting model function~$f$ is state invertible. Moreover, since the value of the cost function~$\ell(x,u)$ is assumed to be measured, the assumption that its image is bounded is not limiting.

\begin{remark}[Soft constraints]
	There may exist a collection of variables  of the system \mbox{$y_c(k)\in \R^{n_y}$} which are subject to (soft) constraints $y_c(k) \in \setY_c$, being $\setY_c$ a closed set. To cope with this case, the stage cost function $\ell(x,u)$ can be used to take into account these constraints by adding a term that penalizes their violation. This term can be thought of as the economic cost of not fulfilling them.
\end{remark}

\subsection{Stabilizing economic MPC} \label{sec:EMPC}

According to the given economic cost function, the optimal equilibrium point is obtained from the solution of the following optimization problem:
\begin{subequations}\label{eq:aux3}
	\beqna
	\label{EqPtoEqF}
	(x_s,u_s) &= &\arg \min_{x, u \in \setU,} \ell(x,u)\\
	\textup{s.t.}&& x=f(x,u). \label{EqPtoEqEf:2}
	\eeqna
\end{subequations}

The economic optimal operation of a system by means of a model predictive control law is a very complex problem that has been thoroughly studied recently. See for instance the excellent survey papers~\cite{RawlingsCDC12,Faulwasser2018economic} and the references therein. For the asymptotic stabilization of economic optimal controllers, the dissipativity property plays an important role. This property has been related to the turnpike property, see~\cite{FaulwasserAUT17} for their relation within continuous-time models in optimal control, and~\cite{grune2016relation,muller2013convergence} for discrete-time models. In this work, this condition is stated in the following assumption:

\begin{assumption} \label{HipDiss}
	The system $f$ is strictly dissipative  with respect to the supply rate $s(x,u)=\ell(x,u)-\ell(x_s,u_s)$, i.e. there exists a storage function $\lambda: \R^n \rightarrow \R$  such that
	\begin{align}
	    \lambda(f(x,u)) - \lambda(x) \leq &- \rho( \|x-x_s\|) - \rho( \|u-u_s\|)\nn\\
	      &+  \ell(x,u)-\ell(x_s,u_s),
	\end{align}
\noindent for certain $\mathcal K$ function $\rho(\cdot)$. It is also assumed that the storage function is locally Lipschitz continuous and bounded below for any admissible trajectory of the system and that $u_s$ lies in the relative interior of $\setU$.
\end{assumption}
{Notice that if this assumption holds, then the optimization problem \eqref{eq:aux3} has an unique solution.}

Two different stabilizing economic model predictive control formulations are considered in this paper:  with and without  terminal equality constraint. An economic MPC with terminal equality constraint is derived from the solution of the following  optimization problem $P^e_N(x(k))$~\cite{AmritARC2011}:
\begin{subequations}\label{eq:PNe}
	\beqna
	\hspace{-0.1cm}\min_{\hat \vu} &&V_{N,e}(x(k),\hat \vu) =\nn\\
	&&\sum_{j=0}^{N-1} \ell(\hat x(j|k),\hat u(j)) \label{eq:PNe1}\\
	\textup{s.t.}&& \hat x(0|k)=x(k) \label{eq:PNe2}\\
	\hspace*{-1cm}&& \hat x(j+1|k)= f (\hat x(j|k),\hat u(j)), j \in \I_0^N\quad \label{eq:PNe3}\\
	&& \hat u(j) \in \setU \label{eq:PNe4}\\
	&& \hat x(N|k) =x_s. \label{eq:PNe5}
	\eeqna
\end{subequations}
\noindent The optimum of this problem is denoted~$V_{N,e}^*(x(k))$.

A more general formulation of the economic MPC without terminal equality constraint can be obtained adding a relaxed terminal constraint and a terminal cost function, leading to the following problem~$P_N^t(x(k))$~\cite{RawlingsCDC12}:
\begin{subequations}\label{eq:PNt}
	\beqna
	\hspace{-0.3cm}\min_{\hat \vu} && V_{N,t}(x(k),\hat \vu) \\ && = \sum_{j=0}^{N-1} \ell(\hat x(j|k),\hat u(j)) + V_f(x(N|k))\quad\\
	\textup{s.t.}&& \eqref{eq:PNe2}-\eqref{eq:PNe4} \nn\\
	&& \hat x(N|k) \in X_f.
	\eeqna
\end{subequations}
\noindent The optimum is denoted~$V_{N,t}^*(x(k))$.

In both cases, a state feedback control law (either $u(k)=\kappa_{\textup{eco}}^e(x(k))$  or~$\kappa_{\textup{eco}}^t(\cdot)$) is obtained applying the solution of the corresponding optimization problem in a receding horizon manner, i.e.  $u(k)=u^*(0;x(k))$.

In this paper, it is assumed that the economic MPC optimization problems satisfy the following condition:

\begin{assumption} \label{as:uniqueness}
	The optimal solution of the problem $P_N^e(x(k))$  ( or $P_N^t(x(k))$) is unique, and the optimal cost function $V_{N,e}^*(\cdot)$ (or $V_{N,t}^*(\cdot)$)  is continuous at the reference~\mbox{$x=x_s$}.
\end{assumption}

The asymptotic stability of the closed-loop system, controlled by the EMPC law with terminal equality constraint, $\kappa_{eco}^e(x)$, was proven in~\citep{AngeliTAC12}. In the case of the EMPC with a relaxed terminal constraint, under a suitable design of the terminal ingredients~$V_f(\cdot)$ and~$X_f$, the resulting control law also stabilizes the system to the optimal equilibrium point. Following~\cite{FaulwasserIECR19}, these conditions are categorized in three different cases as follows:
\begin{itemize}
    \item[TC1]  Terminal inequality constraint~\cite{AmritARC2011}. Considering a locally stabilizing control law, a suitable function $V_f(\cdot)$ and an invariant set $X_f$ satisfying a set of conditions. These ingredients can be determined from the linearized model of the plant at the optimal equilibrium point.
    \item[TC2] Terminal cost function and no terminal constraint~\cite{FaulwasserAUT17}. Taking $X_f=\R^n$, a linear terminal cost function $V_f(x)=\eta_f^Tx$, where $\eta_f$ is the Lagrange multiplier corresponding to constraint (\ref{EqPtoEqEf:2}) in optimization problem (\ref{eq:aux3}), and a sufficiently long prediction horizon $N$.
    \item[TC3]  No terminal ingredients~\cite{GruneLIB17}. Taking $X_f=\R^n$, $V_f(x)=0$ and a sufficiently long prediction horizon, the resulting controller is practically asymptotically stable and the ultimately bound depends inversely on the prediction horizon.
\end{itemize}

For the cases TC2 and TC3, additional technical assumptions are necessary, such as exponential reachability of the equilibrium point, regularity of the steady-state optimization problem and some controllability  conditions of the linearized model at the equilibrium point \cite{FaulwasserNMPC18}.

All these controllers are recursively feasible and stabilize the closed-loop system providing an economically optimal closed-loop trajectory, such that
\begin{equation}
\lim_{T \rightarrow \infty}\frac{1}{T}\sum_{k=0}^{T-1}\ell(x(k),\kappa_{\textup{eco}}(x(k)))\leq \ell(x_s,u_s). \label{EcoPropt}
\end{equation}

Note that in addition to the averaged performance defined here, the literature also contains estimates for transient or non-averaged optimality of economic MPC schemes (see, e.g.,~\cite{GruneLIB17}), which may be interesting to consider in the oracle-based approach proposed in this paper. 



\subsection{Oracle-based economic MPC}
The standard economic MPC formulations presented in this section require knowledge of the prediction model of the plant and the measurement or estimation of the current state. However, there may exists scenarios in which this information is not accessible, due e.g. to privacy or security conditions of the plant. 

The main objective of this work is to design an economic predictive controller that stabilizes the plant and minimizes its economic performance satisfying (\ref{EcoPropt}), under the assumption that the model of the plant is not known and that the only measurement of the plant is the value of the economic cost  at each sampling time. 

Using a database of past inputs and economic cost trajectories, a function used to predict the evolution of the
cost will be obtained. This function is called an oracle, since it allows us to forecast the economic performance of the plant. Once that this oracle is obtained, a suitable predictive controller will be designed  based only in the available measurement of the current economic cost function.


{In the following section, the existence of an oracle is studied. Then, oracle-based economic predictive controllers will be presented, and their properties analyzed.}

\section{The oracle}\label{sec:oracle}

In this section, the procedure to obtain an oracle from past input and performance trajectories, and the conditions under which such an oracle exists are presented. The oracle proposed has the form of a nonlinear auto-regressive model with exogenous signals (NARX), which has been extensively used in nonlinear systems identification~\cite{LeontaritisIJC85}, defined by the  following nonlinear difference equation~\footnote{~We may sometimes aggregate the notation of the cost as $\ell(k)=\ell(x(k),u(k))$.}
\beqn
\hat \ell(k)=\Or(z(k),u(k)),
\label{EqOracle}
\eeqn
where $\hat \ell(k)$ is the estimated cost at sampling time $k$, and $z(k)$ is a vector  given by the following collection of past inputs and costs:
\beqna\label{eq:z}
 z(k)&=&(\ell(k-1),\cdots, \ell(k-n_a),\nn\\
 &&\;u(k-1),\cdots,u(k-n_b)),
\eeqna
\noindent for some memory horizons~$n_a,n_b\in\mathbb N$. The vector $z(k)$ can be regarded as the state of the oracle, and its dimension  is  $n_z=n_a+m \cdot n_b$. It follows that the oracle is a function $\Or: \R^{n_z}\times \R^{m}  \rightarrow \R$, since the economic cost is a real number.

Notice that the value of the economic cost function at time instant~$k$,~$\ell(k)$,  depends in general on the value of the input at the same time instant,~$u(k)$, leading to an inner feed-forward structure. This implies that the state vector $z(k)$ can only depend on the sequence of past costs up to $k-1$, that is, \mbox{$\ell(k-1), \cdots, \ell(k-n_a)$}. Then, the state feedback controller to be designed with the form
$u(k)=\kappa_{\textup{MPC}}(z(k))$ is such that the current control action $u(k)$ depends on the information of the plant available up to $k-1$, given the set-up defined by~(\ref{EqOracle}) and~(\ref{eq:z}).

The conditions under which a system can be described as a NARX have been widely studied during the last 30 years. One of the first results on this topic was given by Sontag~\cite{SontagSJCO79}, relating
the existence of this model to the observability property of the system. Later, Chen and Billings \cite{ChenIJC89} proved that local NARX models can be obtained if the system is locally observable. A comprehensive study of this problem, for local and global estimators, was presented by Levin and Narendra \cite{LevinCAP97}.

In \cite[Theorem~3]{LevinCAP97} it was proven that if the linearised model at the equilibrium point $(x_s,u_s)$ is observable, then the dynamics of the system can be locally described by a NARX model. When a global NARX model is required, the strong observability property must be ensured, from the linearised system, for instance.  However, in \cite[Theorem~6]{LevinCAP97}, the authors proved that only generic observability is necessary, which is a much weaker property.  Since these results are a keystone of this paper, the main results used in~\cite{LevinCAP97} have been summarized in the appendix, including an extra corollary to relate the approach to the case presented in this paper.

Next, for the sake of clarity, we introduce the definition of Morse functions.
\begin{definition}[Morse function]\label{def:morse}
A function~$m(x,u)$ is said to be Morse in~$x$ if for all $(\tilde x,\tilde u)$ where the gradient $\nabla_x m(\tilde x,\tilde u)=0$, the Hessian matrix $\nabla_{xx} m(\tilde x,\tilde u)$ is non-singular~\cite{edelsbrunner2010computational}.
\end{definition}
In order to derive the existence of an oracle of the economic cost, the following assumption must be fulfilled.
\begin{assumption}\label{as:Morse}
The economic cost function $\ell(x,u)$ is a Morse function in~$x$.
\end{assumption}

In virtue of  Corollary \ref{ap:CorNARX} (presented in the appendix), the conditions required to ensure the existence of the oracle can be derived:

\begin{theorem}[Existence of the oracle] \label{ThmOracle}
	Consider that Assumptions~\ref{HipCont} and \ref{as:Morse} hold and that the horizons $n_a$ and $n_b$ are larger than or equal to $2n$. Let $\vu_z(k) \in \setU^{n_b+1}$ be the sequence of $n_b+1$ last inputs applied to the system \eqref{EqMdl} up to sample time $k$. Then, for every $\epsilon>0$ there exists a set of sequences of $n_b+1$ inputs $\setU_{no}\subseteq \setU^{n_b+1}$ of measure $\mu(\setU_{no})<\epsilon$ such that:
	\begin{enumerate}
	    \item There exists a continuous oracle function \eqref{EqOracle} that describes the system~\eqref{EqMdl}, i.e.
	    $$\ell(k)=\Or(z(k),u(k)),$$
	    for any $ \vu_z(k) \in \setU_o=\setU^{n_b+1} \setminus \setU_{no}$.
	    
	    \item There exists a continuous oracle function \eqref{EqOracle} such that for any admissible state $z(k)$ and any admissible inputs $u(k) \in \setU$,
	    $$ |\ell(k)-\Or(z(k),u(k))| \leq \epsilon$$ 
	\end{enumerate}
	
\end{theorem}
From this theorem it follows that under mild conditions on the measured economic cost, an oracle can be constructed to precisely predict the expected evolution of the cost for almost  every sequence of control inputs. Besides, it is also proved that assuming certain (arbitrarily small) description error, an oracle can be found for any admissible input.

The procedure to derive the oracle can be obtained using estimation and learning theory methods. There exists a number of
methods capable of approximating the real function (the so-called ground truth function) from possibly noisy sampled data, such as support vector machines, neural networks or direct weight optimization \cite{roll2005general}.
Recently, other non-parametric methods such as Gaussian processes \cite{rasmussen2006gaussian} or kinky inference \cite{manzano2019output,calliess2014_thesis} have gained a lot of attention, thanks to their capability to provide estimations of the prediction error.

As the model order $n$ may not be known a priori, the parameters of the chosen estimator, including the memory horizons~$n_a$ and~$n_b$, can be calculated from a database of historical inputs and costs, namely the \emph{training data set}. In addition, a different collection of data points is used for validation of the proposed estimator. This cross-validation methodology allows one to derive the best structure of the estimator, as well as the best horizons $n_a$ and $n_b$, from the real data.

{Notice that the requirement for the horizons~$n_a,n_b$ to be lager than~$2n$ can be theoretically addressed by considering that the oracle function depends on a sequence of at least~$2n$ past inputs and outputs, where some of its components may not affect to the output of the oracle.}

\section{Oracle-based economic predictive control}\label{sec:oempc}

In this section, the proposed oracle-based predictive controller is presented, and its stability and optimality properties are studied, under the assumption of perfect estimation. 


Assuming that an oracle is available, the system can be described by an state-space prediction model as follows:
\begin{subequations}\label{EqPredictor}
\begin{eqnarray}
\hat z(j+1|k)&=&\hat F(\hat z(j|k),\hat{u}(j)), \label{EqPredictor1} \\
\hat \ell(j|k)&=&\Or(\hat z(j|k),\hat u(j)), \label{EqPredictor2}
\end{eqnarray}
\end{subequations}
where the predicted state $\hat z(j|k) \in \R^{n_z}$ is given by
\begin{eqnarray} 
	\hat z(j|k)&=&\big(\hat \ell(j-1|k), \cdots, \hat \ell(k),\cdots ,\nn\\
	&&\;\ell(k+j-n_a),\;\hat{u}(j-1), \cdots,\nn\\ 
	&& \;\hat{u}(0),\cdots,u(j-n_b)\big),
\end{eqnarray}
for $j\geq 1$. This state includes measured costs $\ell$ and inputs $u$ if $n_a\geq j$ or $n_b\geq j$  respectively,
and only estimated values $\hat{\ell   }$ or $\hat{u}$ otherwise.


Thus, the prediction model  is
\begin{eqnarray*}
	\hspace{-0.5cm}\hat F(\hat z(j|k),\hat u(j))&=&\big(\Or(\hat z(j|k),\hat u(j)),\hat \ell(j-1|k),\nn\\
	&&\;\cdots,\hat \ell(1|k),\ell(k), \cdots,\nn\\
	&&\;\ell(k+j-n_a+1),\hat u(j),\nn\\
	&&\;\cdots,\hat u(j-n_b+1)\big).
\end{eqnarray*}


Using this oracle-based prediction model,  economic predictive controllers will be  presented, based on both terminal equality constraint and relaxed terminal inequality constraint designs. The stability of the proposed controllers will be analyzed in two steps, as it is customary in the predictive control field. First, nominal stability is studied, assuming that the predictions (provided by the oracle) are exact. Next, in the following section, stability under prediction mismatches is analyzed.



\begin{assumption}[Exact oracle] \label{HipNomCond}
    The sequence of the last $n_b+1$ inputs  applied to the system at every sample time $k$, $\vu_z(k)$, is contained in the set $\setU_o$ defined in Theorem~\ref{ThmOracle} and the oracle exactly forecasts the cost, i.e.~\mbox{$\Or(z(k),u(k))=\ell(k)$}.
\end{assumption}


From Theorem \ref{ThmOracle}, we have that in practice this hypothesis could hold true for any sequence of inputs with probability practically equal to 1, since the set of non-observable sequence of inputs is of measure zero.

In order to derive the economic MPC based on the oracle, the economically optimal equilibrium point (also based on the oracle) must be calculated first. To this end, the following optimization problem, similar to (\ref{eq:aux3}), has to be solved:
\begin{subequations}\label{eq:ssto}
	\beqna
	(u_s,\ell_s) & =  & \arg \min_{u \in \setU,\ell} \ell\\
	\text{s.t}.&& \ell=\Or(z,u) \\
	&& z=(\ell, \ldots, \ell, u,\ldots,u). \label{eq:ssto3}
	\eeqna
\end{subequations}

In the following theorem, it is proven that the optimal equilibrium point derived from this optimization problem is equivalent to the solution of~(\ref{eq:aux3}).

\begin{theorem}\label{thm:ssto}
If Assumptions \ref{HipCont}-\ref{HipNomCond} hold, then the optimal equilibrium point derived from the solution of~\eqref{eq:ssto}, based on the oracle, is equal to the optimal equilibrium point of the plant, derived from  \eqref{eq:aux3}.
\end{theorem}
\begin{proof}
Since the state $z(k)$ contains the last $2n$ measures of the measured input and output, from statement 2 of Theorem \ref{ap:ThmNARX} and  Corollary \ref{ap:CorNARX} (both presented in the appendix),  there exists a continuous bijective function  $\Phi_z(\cdot)$ such that 
\begin{equation} 
    (x(k-2n),\vu_z(k)) = \Phi_z(z(k)), \label{eq:Phiz}
\end{equation}
where  $\mathbf u_z(k)$ denotes the sequence of past inputs contained in $z(k)$, such that $\vu_z(k)=(u(k-1),\cdots, u(k-2n))$. Let denote~$\Phi(\cdot)$ the continuous function $x(k-2n)=\Phi(z(k))$, given by the first $n$ components of the map~$\Phi_z(\cdot)$.
 
 Let $(\bar x_s,\bar u_s)$ be the optimizer of  \eqref{eq:aux3}, and $\bar \ell_s$ its optimal economic cost function. Let $(\ell_s,u_s)$ be the optimizer of \eqref{eq:ssto} and let $z_s$ be the optimal steady state of the oracle-based model given by \eqref{eq:ssto3}.  Since ($\bar x_s, \bar u_s)$ is an equilibrium point of~\eqref{EqMdl}, then defining 
$$\bar z_s=(\bar \ell_s, \ldots, \bar \ell_s, \bar u_s,\ldots,\bar u_s),$$ 
we have that  $(\bar z_s,\bar u_s)$ is an equilibrium point of~\eqref{EqPredictor}. Thus, $(\bar \ell_s,\bar u_s)$ is a feasible solution of~\eqref{eq:ssto}, which means that $\bar \ell_s \geq \ell_s$.

Since the map $\Phi(\cdot)$ is continuous, the condition $z(k)=z(k+1)$ implies that $x(k-2n)=x(k-2n+1)$. That is, an equilibrium point of \eqref{EqPredictor} corresponds to an equilibrium of \eqref{EqMdl}, and vice versa. Indeed, for a given $u_s$, both equilibrium points are related by the continuous map $x_s=\Phi(z_s)$. 

Therefore, the pair $(\Phi(z_s),u_s)$ is a feasible solution of \eqref{eq:aux3}, and its optimal cost function is $\ell_s$, which leads to $\ell_s \geq \bar \ell_s$, and consequently, $\ell_s=\bar \ell_s$.

From Assumption~\ref{HipDiss}, we derive that the solution of  \eqref{eq:aux3} is also unique, which means that $\bar u_s=u_s$, and $\bar x_s = \Phi(z_s)=x_s$.
\end{proof}

In what follows, the economic model predictive controllers based on the proposed oracle are presented.

\subsection{Oracle-based economic MPC with terminal equality constraint}

This formulation is one of the simplest stabilizing designs of the economic MPC. If the state-space system model were available, the control law could be derived from the optimization problem $P_N^e(x(k))$.  In  case that the oracle is used as prediction model, we present a formulation based uniquely in the available information, whose control law is proven to be equivalent to the control law derived from the state-space prediction model. 

The proposed oracle-based economic MPC with terminal equality constraint is derived from the solution of the following optimization problem $P_{\Or}^e(z(k))$:
\begin{subequations}\label{eq:PO}
	\beqna
	\hspace{-0.7cm}\min_{\hat \vu} && V_{\Or,e}(z(k),\hat \vu) =  \sum_{j=0}^{N-1} \hat \ell(j|k)\\
	\hspace{-0.7cm}\textup{s.t.}&& \hat z(0|k)=z(k)\\
	&& \hat z(j+1|k)= \hat F(\hat z(j|k),\hat u(j)), j \in \I_{0}^{N_p-1}\quad\\
	&& \hat \ell(j|k)=\Or(z(j|k),u(j))\\
	&& \hat u(j) \in \setU\\
	&& \hat u(j) = u_s,\,\forall j \in \I_{N}^{N_p-2}\label{eq:tc:po:u}\\
	&& \hat \ell (j|k)=\ell_s ,\,\forall j \in \I_{N}^{N_p-1},\label{eq:tc:po:y}
	\eeqna
\end{subequations}
where $N_p=N+\max(n_a,n_b)$ and $(\ell_s,u_s)$ is the optimal solution of \eqref{eq:ssto}.

 Note that the terminal constraint, defined by~(\ref{eq:tc:po:u}) and~(\ref{eq:tc:po:y}), requires that both the input and the cost maintain certain constant value for a given period of time, defined by the memory horizons of the NARX model. 
The control law is then given by
$$u(k)=\kappa_\Or^e(z(k))=\hat u^*(0;z(k)).$$

In order to compare the oracle-based controller and its state-space model counterpart, we include the following definition.

\begin{definition}[Consistent states]\label{def:consistent}
For a given trajectory of the system, we say that the state $z(k)$ of the model \eqref{EqPredictor} is \emph{consistent} with the state $x(k)$ of the model~\eqref{EqMdl} if~\eqref{eq:z} holds, i.e., if
\begin{eqnarray*}
z(k)=\big(\ell(x(k-1),u(k-1)), \cdots, \ell(x(k-n_a),\\
u(k-n_a)),u(k-1),\cdots,u(k-n_b)\big).
\end{eqnarray*}
\end{definition}


In the following theorem, it is stated that the oracle-based economic MPC is equivalent to the model-based economic MPC  under nominal conditions, and hence it renders the controlled system asymptotically stable and minimizes the economic performance.

\begin{theorem}\label{th:eq}
	Consider that Assumptions \ref{HipCont}-\ref{as:Morse} hold and assume that the initial state $x(0)$ is such that  $P_N^e(x(0))$ is feasible. Then, under Assumption \ref{HipNomCond}  the optimization problem $P_{\Or}^e(z(0))$  is feasible, where $z(0)$ is the state of the oracle consistent with $x(0)$, and the evolution of the system controlled by the economic control law $\kappa_\Or^e(z(k))$ 
	is equal to the one resulting from the control law derived from $P_N^e(x(k))$, that is, $\kappa_\Or^e(z(k))=\kappa_{\textup{eco}}^e(x(k))$.
\end{theorem}
\begin{proof}

First, take into account that in virtue of Theorem \ref{ThmOracle} and Assumption~\ref{HipNomCond}, if~$x(k)$ and~$z(k)$ are consistent, then for a certain sequence of future control inputs $\vu=(u(0),\cdots, u(N-1))$, the predicted trajectories using the plant model~\eqref{EqMdl} and the oracle \eqref{EqPredictor} are such that  
\begin{equation}
\label{eq:consistencia}
\ell(x(j|k),u(j)) = \hat \ell(j|k)=\Or(z(j|k),u(j)).
\end{equation}
    
Assume that $x(k)$ is such that  $P_N^e(x(k))$ is feasible, then there exists a sequence of $N$ inputs $\tilde \vu$ that steers the system to the equilibrium point $(x_s,u_s)$. From Theorem \ref{thm:ssto}, we have that the optimal solution of \eqref{eq:ssto}, $\ell_s$, is equal to $\ell(x_s,u_s)$.  From the consistency between the models \eqref{eq:consistencia}, it is inferred that $\hat \ell(N|k)=\ell_s$. Since $x(N|k)=x_s$, applying subsequently $u_s$, the system remains in the equilibrium point $(x_s,u_s)$ and then $\hat \ell(j|k)=\ell(x_s,u_s)=\ell_s$ for $j\in \mathbb I_N^{N_p}$. Therefore, $P_\Or^e(z(k))$ is feasible.

From this, we can infer that $P_\Or^e(z(0))$ is feasible given that $P_N^e(x(0))$ is assumed to be feasible.


In order to prove the equivalence between the control laws, it is demonstrated that (i) the optimal solution of $P_N^e(x(k))$ is a solution of $P_\Or^e(z(k))$ and (ii) the opposite.
	
(i) It was proven that if $P_N^e(x(k))$ is feasible, then the optimal solution of $P_N^e(x(k))$, {$\vu^*_x$}, is a feasible solution of  $P_\Or^e(z(k))$, and then the costs are such that~$V_{\Or,e}(z(k),\vu^*_x) \geq V_{\Or,e}^*(z(k))$. Taking into account \eqref{eq:consistencia}, we have that $V_{\Or,e}(z(k),\vu^*_x)=V_{N,e}^*(x(k))$, and then $V_{N,e}^*(x(k)) \geq V_{\Or,e}^*(z(k))$.
	
	(ii) On the other hand, in order to prove the opposite, notice that the set of constraints \eqref{eq:tc:po:u} and \eqref{eq:tc:po:y} is equivalent to force that $z(N_p|k)=z_s$.


	Hence, as the optimal solution of $P_\Or^e(z(k))$, $\vu^*_z$ leads to $z(N_p|k)=z_s$, recalling function $\Phi(\cdot)$ defined in the proof of Theorem \ref{thm:ssto}, we have that
	$$x(N|k)=x(N_p-2n|k)=\Phi(z(N_p|k))=\Phi(z_s)=x_s,$$
	and therefore {$\vu^*_z$} is a suboptimal solution of $P_N^e(x(k))$. From \eqref{eq:consistencia}, we have that~$V_{N,e}(x(k),\vu^*_z)=V_{\Or,e}^*(z(k))$, so $V_{\Or,e}^*(z(k))=V_{N,e}(x(k),\vu^*_z)\geq  V^*_{N,e}(x(k))$.
	
	Then, it is straightforward that $V^*_{N,e}(x(k))=V^*_{\Or,e}(z(k))$, so the solution is the same for both problems, i.e. $\vu^*_x=\vu^*_z$, provided that Assumption~\ref{as:uniqueness} holds.
\end{proof}

From this theorem we conclude that the proposed oracle-based economic control law $\kappa_{\Or,e}(z(k))$ provides the same domain of attraction, the same optimality properties and the same evolution that the state-space law $\kappa_{N,e}(x(k))$. This demonstrates that an oracle of the economic cost function suffices to derive economic predictive controllers under nominal conditions.

	%

\subsection{Oracle-based economic MPC without  terminal equality constraint}

A stabilizing economic model predictive control design resorting on terminal equality constraint is simple and direct, but the resulting controller might exhibit a small domain of attraction and lack of inherent robustness \cite{GrimmAUT04}.  In this section, stabilizing oracle-based economic predictive control laws are derived extending  the existing terminal conditions design summarized in Section~\ref{sec:EMPC} to the oracle-based framework. This  design is only based on the input-measured cost available information, and the stability of the resulting  controller is proven.

Taking a stabilizing design based on terminal ingredients (cf. TC1 in Section~\ref{sec:EMPC}), the optimization control problem $P_{\Or}^t(z(k))$ is formulated as follows, where a prediction horizon $N_p$ larger than the control horizon $N_c$ can be considered:
\begin{subequations}\label{eq:POt}\allowdisplaybreaks
	\beqna
	\hspace{-0.6cm}\min_{\hat \vu} && V_{\Or,t}(z(k),\hat \vu) =\nn\\
	&&\sum_{j=0}^{N_p-1} \hat \ell(j|k) + V_{\Or,f}(\hat z(N_p|k))\\
	\hspace{-0.6cm}\textup{s.t.}&& \hat z(0|k)=z(k)\label{eq:PO1} \\
	&& \hat z(j+1|k)= \hat F(\hat z(j|k), \hat u(j)), j \in \I_{0}^{N_p-1}\\
	&& \hat \ell(j|k)=\Or(z(j|k),\hat u(j))\\
	&& u(j) \in \setU\\
    && \hat z(i+1|k)= \hat F(\hat z(i|k), \kappa_f(z(i|k))), \nn\\
    &&i \in \I_{N_c}^{N_p-1}\\
	&& \hat \ell(i|k)=\Or(z(i|k),\kappa_t(z(i|k)))\\
	&& \kappa_f(z(i|k)) \in \setU\\
	&& z(N_p|k)\in \setZ_f. \label{eq:PO8}
	\eeqna
\end{subequations}

The control law is implicitly obtained applying the optimal sequence of control inputs in a receding horizon manner $u(k)=\kappa_{\Or}^t(z(k))=\hat u^*(0|k)$.

In order to derive asymptotic stability of the economic MPC, the following assumption is made~\citep{AmritARC2011}, in which the terminal cost function $V_{\Or,f}$ and the terminal set $\setZ_f$ must satisfy the following conditions:
\begin{assumption} \label{as:TC1}
     $V_{\Or,f}$ is continuous at $z=z_s$ and there exists a terminal control law $u=\kappa_f(z)$ such that for all $z \in \setZ_f$, $\kappa_f(z) \in \setU$, $z^+ = \hat F(z,\kappa_f(z)) \in \setZ_f$  and 
    $$ V_{\Or,f}(z^+) - V_{\Or,f}(z) \leq -\Or(z,\kappa_f(z)) + \ell(x_s,u_s).$$
\end{assumption}

If the linearization of the oracle-based autoregressive model \eqref{EqOracle} at the equilibrium point given by \eqref{eq:ssto} is stabilizable, then the terminal ingredients can be calculated from the linearization, as proposed in \citep{AmritARC2011}, yielding a linear terminal control law $\kappa_f(z)=K (z-z_s)+u_s$, a quadratic terminal cost~$V_{\Or,f}(z)=\frac{1}{2} z^T Q_fz + q_f^T z$, and a terminal constraint $\setZ_f=\{\bar z: \bar (z-z_s)^T P \bar (z-z_s) \leq \alpha \}$.



We are now in position to study the stability of the proposed control law.

\begin{theorem} \label{ThmASTC}
	Consider that Assumptions \ref{HipCont}-\ref{as:TC1} hold. If the initial state $z(0)$ is such that  $P_\Or^t(z(0))$ is feasible, then the equilibrium point given by \eqref{eq:ssto} is asymptotically stable for the system \eqref{EqMdl} controlled by the economic control law $\kappa_\Or^t(z(k))$ and its trajectory satisfies the constraints.
\end{theorem}

\begin{proof}

First, we show that there exists a map that relates the real state of the system $x(k)$ and the state of the oracle-based model $z(k)$. From Assumption \ref{HipCont} we have that the  model function $x \mapsto f(x,u)$ defines a diffeomorphism in $x$ for a given $u$. Therefore, there exists a continuous and bijective map $\mathcal F$ such that 
 \begin{equation} \label{eq:mapF}
      (x(k),\vu_z(k))=\mathcal F(x(k-2n),\vu_z(k)).
 \end{equation}

Taking into account that $z(k)$ contains the sequence $\vu_z(k)$ and the maps~\eqref{eq:Phiz} and~\eqref{eq:mapF}, we can state that  $\mathcal G=\mathcal F \circ \Phi_z$ is a continuous and bijective map such that
\begin{equation} \label{eq:mapG}
      (x(k),\vu_z(k))=\mathcal G(z(k)).
 \end{equation}
 
Based on this map,  it is shown that the system \eqref{EqOracle} is dissipative with respect to the output of the oracle.  Consider the continuous function $x(k)=G_x(z(k))$ given  by the first $n$ components of map $\mathcal G (\cdot)$ introduced in \eqref{eq:mapG}, and  define the storage function $\lambda_z(z)=\lambda(G_x(z))$.
 
 Then, from Assumptions \ref{HipDiss} and \ref{HipNomCond} and Theorem~\ref{thm:ssto}, we have that
 \begin{align*}
\lambda_z(z(k+1))&- \lambda_z(z(k))= \lambda(x(k+1))- \lambda(x(k))\\
 \leq & \,\ell(x(k),u(k))-\ell(x_s,u_s) \\
  &+ \rho(\|x(k)-x_s\|) +\rho(\|u(k)-u_s\|)\\
 = \,&\Or(z(k),u(k)) - \ell_s\\
  &+  \rho(\|x(k)-x_s\|) + \rho(\|u(k)-u_s\|).
 \end{align*}

Stability of the proposed controller can be stated using the rotated stage cost function and terminal cost function, as proposed in \citep{AmritARC2011}:
\beqnan
\hspace{-0.6cm}L_r(z,u)&=& \Or(z,u) + \lambda_z(z)- \lambda_z(\hat F(z,u)) -\ell_s,\\
\hspace{-0.6cm}V_{\Or,fr}(z)&=&V_{\Or,f}(z) + \lambda_z(z)  - V_{\Or,f}(z_s) - \lambda_z(z_s).
\eeqnan
Notice that the rotated stage cost~$L_r$ satisfies the following condition:
$$ L_r(z,u) \geq \rho(\|x-x_s\|) + \rho(\|u-u_s\|), $$
where $x=G_x(z)$.

Now, we define the  optimization problem $P_{\Or,tr}(z(k))$ based on the rotated ingredients as:
\begin{subequations}\label{eq:POtr}
\begin{eqnarray}
\min_{\hat \vu}&&  V_{\Or,tr}(z(k),\hat \vu) = \nn\\
&&\;\sum_{j=0}^{N_p-1} L_r(j|k) + V_{\Or,fr}(\hat z(N_p|k))\\
\mathrm{s.t.}&& \eqref{eq:PO1}-\eqref{eq:PO8}\\
&&L_r(j|k) = \hat \ell(j|k) + \lambda_z(\hat z(j|k))\nn\\
&&\;- \lambda_z(\hat z(j+1|k)) - \ell_s.
\end{eqnarray}
\end{subequations}

Following \citep{AmritARC2011}, it can be proven that
\beqnan
\hspace{-0.7cm}V_{\Or,tr}(z(k),\hat \vu) &=&  \sum_{j=0}^{N_p-1} L_r(j|k) + V_{\Or,fr}(\hat z(N_p|k)) \\
&=&  V_{\Or,t}(z(k),\hat \vu) + \lambda_z(z(k)) -N_p\ell_s \nn\\
&&\; - V_{\Or,f}(z_s) - \lambda_z(z_s),
\eeqnan
and then the optimal solution of~$P_{\Or,tr}(z(k))$ is equal to the optimal solution of~$P_{\Or,t}(z(k))$. Therefore, the receding horizon control laws derived from these two optimization problems are identical.

Next, asymptotic stability of the proposed economic MPC is addressed. First,it can be shown that the optimal cost function 
$V^*_{\Or,tr}(z)$ is such that
\begin{align} \label{eq:Tm4cotarho}
    V^*_{\Or,tr}(z(k)) &\geq L_r(z(k),\hat u^*(0|k)) \nn\\
    &\geq \rho(\|x(k)-x_s\|) + \rho(\|u(k)-u_s\|),
\end{align}
for all feasible $z(k)$. 

In order to prove asymptotic stability, we introduce the following Lyapunov function candidate:
\begin{equation}
    W(z(k))=\sum_{j=k-2n}^{k} V^*_{\Or,tr}(z(j)),
\end{equation}
and next we will prove that it satisfies the sufficient conditions to derive asymptotic stability of the controlled system.

From \eqref{eq:Tm4cotarho}, we have that
\begin{align} 
    W(z(k))&=\sum_{j=k-2n}^{k} V^*_{\Or,tr}(z(j))\nn\\
    &\geq \sum_{j=k-2n}^{k} \rho(\|x(j)-x_s\|) + \rho(\|u(j)-u_s\|). \label{eq:Tm4cotarho2}
\end{align}
Denoting the sequences $\vx_{2n}(k-2n)=(x(k),\cdots, x(k-2n))$,
and $\vu_{2n}(k-2n)=(u(k),\cdots, u(k-2n))$ and $\vx_{2n}^s$ and $\vu_{2n}^s$ their steady values counterparts at $(x_s,u_s)$, we can see that the right hand side of \eqref{eq:Tm4cotarho2} is a $\mathcal K$-function  of the sequence of $\vx_{2n}(k-2n)- \vx_{2n}^s$ and $\vu_{2n}(k-2n) - \vu_{2n}^s$. From similar arguments to the derivation of map $\mathcal F$ previously presented in \eqref{eq:mapF}, there exists a bijective map $\mathcal F_{2n}$ such that
\begin{align}
    (\vx_{2n}(k-2n),&\vu_{2n}(k-2n))=\nn\\
    &\mathcal F_{2n} (x(k-2n),\vu_{2n}(k-2n)).
\end{align}
Taking into account \eqref{eq:Phiz}, and denoting $\mathcal H_{2n}= \mathcal F_{2n} \circ \Phi_z$, we have that $\mathcal H_{2n}$ is a bijective map such that
\begin{equation}
    (\vx_{2n}(k-2n),\vu_{2n}(k-2n))=\mathcal H_{2n}(z(k)),
\end{equation}
and then $(\vx_{2n}(k-2n),\vu_{2n}(k-2n)) = (\vx_{2n}^s,\vu_{2n}^s)$ if and only if $z(k)=z_s$. Therefore, there exists a $\mathcal K$-function $\beta_1$ such that

\begin{align} \label{eq:Tm4cotarho3}
    W(z(k))\geq& \sum_{j=k-2n}^{k} \rho(\|x(j)-x_s\|) + \rho(\|u(j)-u_s\|) \nn\\
    \geq& \beta_1(\|z(k)-z_s\|).
\end{align}

On the other hand, we have that $V^*_{\Or,tr}(z) \leq V_{\Or,fr}(z)$ for all $z-z_s  \in \setZ_f$ (see Proposition~2.18 in~\citep{RawlingsLIB17}), which means that the optimal cost function $V^*_{\Or,tr}(z)$ is continuous at $z_s$ and locally bounded for all $z$ such that  $z-z_s  \in \setZ_f$. Consequently $W(z)$ is also continuos at $z_s$ and locally bounded and in virtue of Proposition B.25 in \citep{RawlingsLIB17Preface}, there exists a $\mathcal K$-function $\beta_2(\cdot)$ such that for all $z-z_s  \in \setZ_f$
$$ W(z) \leq \beta_2(\|z-z_s\|).$$

Finally, following \cite{AmritARC2011} we have that
\begin{align}
    V^*_{\Or,tr}(z(k&+1)) -  V^*_{\Or,tr}(z(k)) \leq - L_r(z(k),\hat u^*(0|k))\nn \\
    &\leq - \rho(\|x(k)-x_s\|) -\rho(\|u(k)-x_s\|).
    \label{eq:DecVNO}
\end{align}

Then, we can state that
\begin{align}
    W(z&(k+1)) -  W(z(k)) \leq\nn\\
    &\sum_{j=k-2n}^{k} - \rho(\|x(j)-x_s\|) -\rho(\|u(j)-x_s\|).
    \label{eq:DecW}
\end{align}
From \eqref{eq:Tm4cotarho3}, we get the last sufficient condition in order to derive asymptotic stability
\begin{equation}
    W(z(k+1)) -  W(z(k)) \leq -\beta_1(\|z(k)-z_s\|).
    \label{eq:DecW}
\end{equation}

Consequently, the system controlled by the control law derived from \eqref{eq:PO} is stable, and it asymptotically steers the state $z(k)$ to the equilibrium point~$z_s$. From~\eqref{eq:Phiz}, this implies that the state of system, $x(k)$, is equivalently steered to the economically optimal steady state  $x_s$. 
\end{proof}

%

\begin{remark}\textbf{(Economic optimality)}
    From equation~\eqref{eq:DecVNO} we can derive that 
    \begin{equation}
    V^*_{\Or,t}(z(k+1)) -  V^*_{\Or,t}(z(k)) \leq - \ell(z(k),\kappa_{\Or}^t(z(k))) - \ell_s.
\end{equation}
and therefore, equation~\eqref{EcoPropt} also holds for the proposed controller.
\end{remark}

Based on the previous stability proof of the oracle-based economic MPC, the stabilizing designs without terminal constraint proposed by \cite{FaulwasserAUT17} and \cite{GruneLIB17} (cf. TC2 and TC3 in Section \ref{sec:EMPC}) can also be extended to the oracle-based case, provided that some additional technical conditions are fulfilled. In this case, the control law is derived from the following optimization problem.

\begin{subequations}\label{eq:POtc}
	\beqna
	\hspace{-0.7cm}\min_{\hat \vu} && V_{\Or,tc}(z(k),\hat \vu) =  \sum_{j=0}^{N-1} \hat \ell(j|k) + \eta^T \hat z(N|k)\nn\\
	&&\\
	\hspace{-0.7cm}\textup{s.t.}&& \hat z(0|k)=z(k) \label{eq:PO31}  \\
	&& \hat z(j+1|k)= \hat F(\hat z(j|k), \hat u(j)), j \in \I_{0}^{N-1}\\
	&& \hat \ell(j|k)=\Or(z(j|k),\hat u(j))\\
	&& u(j) \in \setU.
	\eeqna
\end{subequations}

Thus, taking a sufficiently large prediction horizon $N$ and an appropriate Lagrange multiplier~$\eta$, the resulting control law asymptotically stabilizes the system \cite{FaulwasserAUT17}. Taking $\eta=0$, then only practical stability can be achieved \cite{GruneLIB17}.

\section{Practical viability of the oracle-based economic MPC} \label{sec:practical}
{In the previous section, asymptotic stability of the oracle-based economic MPC was proven under the nominal conditions stated in Assumption~\ref{HipNomCond}. The exact prediction considered in this hypothesis may hold if every applied sequence of inputs rendered the system observable, and if the available oracle NARX model were exact.}

{In this section we study the case in which the available oracle is inexact,  providing approximate predictions with a bounded estimation error. We consider that the two previous assumptions may not hold, which is likely in practice, either due to non-observable input sequences or because of oracles with estimation errors.

In practice, the oracle is determined from input-output historical data, generated with an appropriate sequence of inputs, for which a bound on the estimation error signal
\begin{equation} \label{eq:EstErr}
    |w(k)|=|\ell(x(k),u(k))- \Or(z(k),u(k))|
\end{equation}
may be estimated. There exists machine learning methods that provide either guaranteed bounds, such as kinky inference \citep{calliess2014_thesis,CanaleIJRNC14}, or probabilistic bounds, such as Gaussian processes~\citep{rasmussen2006gaussian}. We assume that we are in the (worst-case) scenario that the obtained bound is only  valid for every sequence of inputs that makes the system observable.  

{In the following proposition, it is proven that an inexact oracle obtained from input-output data, possibly derived using observable sequences of inputs, with a bounded estimation error, also provides estimations with a bounded error for any sequence of inputs, even if they are not observable.}



\begin{proposition} \label{PropErrorOracle}
{Consider that Assumptions \ref{HipCont} and \ref{as:Morse} hold} and that the oracle function $\Or(\cdot,\cdot)$ is a {continuous} function such that the estimation error~\eqref{eq:EstErr} is bounded by~$\tilde \mu$ for every sequence of inputs for which the system is observable. Then, for every admissible sequence of inputs, the estimation error is bounded by~ $\mu>\tilde \mu$. 
\end{proposition}

\begin{proof}
From Theorem \ref{ap:ThmNARX}, derived in the appendix from~\citep[Theorem 6]{LevinCAP97}, we have that  for any (arbitrarily small) constant $\epsilon>0$, the set of sequence of inputs for which the system is not observable, $\setU_{no}$, is of measure $\epsilon$, i.e. arbitrarily small.
Let assume that the state of the system is $x(k-2n)$ and a sequence of $2n$ inputs $\vu=(u(k-1),\cdots, u(k-2n)) \in \setU_{no}$ is applied, evolving the system with a sequence of outputs $\mathbf l=(\ell(k-1),\cdots, \ell(k-2n))$. 

Since $\vu \in \setU_{no}$, there exists a sequence $\tilde \vu \not \in \setU_{no}$ such that $\|\tilde \vu - \vu\|\leq \theta_u(\epsilon)$, being $\theta_u$ a $\mathcal K$-function. Let $\tilde l$, $\tilde x(k)$ be the corresponding sequence of outputs and the state derived from the application of $\tilde  \vu$ at $x(k-2n)$. From the smoothness of the transition map $f$ and of the output map $\ell$, we can derive that there exists a couple of $\mathcal K$-functions $\theta_y$ and $\theta_x$ such that $\| \tilde x(k) - x(k) \| \leq \theta_x(\epsilon)$ and
$\|\mathbf{\tilde l} -  \mathbf l \| \leq \theta_y(\epsilon)$. Let define $z(k)$ as the regressor state that is composed of $\mathbf l$ and $\vu$, and $\tilde z(k)$ as the one composed of $\mathbf{\tilde l}$ and $\tilde \vu$. Then, the estimation error can be bounded as follows:
\beqnan
 |w(k)| &=& |\ell(x(k),u(k))-\Or(z(k),u(k))|\\
 &\leq& |\ell(x(k),u(k)) -  \ell(\tilde x(k),u(k))| \\
 && + | \ell(\tilde x(k),u(k)) - \Or(\tilde z(k),u(k))|\\
 && + |\Or(\tilde z(k),u(k)) - \Or(z(k),u(k))|.
\eeqnan

From the smoothness of $\ell(x,u)$, we have that this function is locally bounded. Then, from \cite[Proposition B.25]{RawlingsLIB17Preface}, we have that for all $x(k) \in \{x: \|x-\tilde x(k)\| \leq \theta_x(\epsilon)\}$ there exists a $\mathcal K$-function $\theta_\ell(\cdot)$ such that
\begin{align*}
|\ell(x(k),u(k)) -  \ell(\tilde x(k),u(k))| &\leq \theta_\ell(|x(k)-\tilde x(k)|) \\
&\leq \theta_\ell(\theta_x(\epsilon)).
\end{align*}
On the other hand, since $\tilde \vu$ is a sequence of inputs that make the system observable, we have that  
$ | \ell(\tilde x(k),u(k)) - \Or(\tilde z(k),u(k))| \leq \tilde \mu$.

By construction, \mbox{$\|z(k)-\tilde z(k)\|\leq \theta_z(\epsilon)$}, where~$\theta_z$ is a~$\mathcal K$-function.{In addition, since the oracle function $\Or(z,u)$ is assumed to be continuous in its arguments and the domains of $z$ and $u$ are compact, the oracle function is uniformly continuous in its domain.  Therefore, we can state that there exists a $\mathcal K$-function $\theta_\Or$ such that
\begin{align*}
    \Or(\tilde z(k),u(k)) - \Or(z(k),u(k))| &\leq \theta_\Or(z(k)-\tilde z(k)) \\
    &\leq \theta_\Or(\theta_z(\epsilon)). 
\end{align*}}

Recapping, we have demonstrated that
$$ |w(k)| = |\ell(x(k),u(k))-\Or(z(k),u(k))| \leq \tilde \mu + \Theta(\epsilon),$$
for a $\mathcal K$-function $\Theta(\cdot)$. Since $\epsilon$ can be taken arbitrarily small, for any $\mu >\tilde \mu$, there exists a suitable $\epsilon>0$ for which the bound on the estimation error holds.
\end{proof}

{
\begin{remark} Notice that the hypothesis that the oracle function is continuous is a mild assumption, since the form of the oracle function is chosen by the user, so it can be designed to satisfy this condition. For instance, oracles derived from Lipschitz interpolation methods, neural networks or  Gaussian processes can be chosen to guarantee the satisfaction of this assumption.
\end{remark} 

}


Next, we will prove the practical viability of the proposed controllers, demonstrating that the closed-loop system is input-to-state stable (ISS) w.r.t.~the model mismatch signal $w(k)$.

It is well known that a nominally stabilizing MPC may exhibit zero-robust\-ness due to the the discontinuity of the optimal solution~\cite{GrimmAUT04}. In \cite[Theorem~4]{LimonNMPC08}, it is proven that if the ingredients of the MPC are uniformly continuous and the constraints on the state in the optimization problem are not active, then the resulting control law is ISS in a region where the constraints on the states are not active.

Although the system to be controlled does not have constraints on the states, the optimization problems of the  proposed oracle-based economic predictive controllers ($P_{\Or}^e$ and $P_{\Or}^t$) have constraints on the states, arising from the terminal constraint. .

In the case of $P_{\Or}^e$, the terminal equality constraint is instrumental for the stability and hence, it cannot be removed without compromising the stability (although, as proven in \citep{GruneLIB17}, practical stability can be achieved for a sufficiently large $N$).

Stability of predictive controllers with terminal cost but without terminal constraint has been analyzed in~\citep{LimonTAC06}. In that paper, it was proven that there exists a level set of the optimal cost function where the MPC without terminal constraint is stabilizing. This level set can be enlarged if the prediction horizon is increased, or if the terminal cost function is weighted. This result was extended in \cite{limon2018nonlinear} to the tracking case, and a prediction horizon larger than a control horizon. As analysed in~\cite{manzano2019output}, larger prediction horizons increase the domain of attraction while barely increasing computational complexity, in contrast to choosing larger control horizons.

Taking these arguments into account, we can derive robust stability of the proposed oracle-based economic MPC. To this aim, we first need to proof the following lemma, which is an extension of~\citep{LimonTAC06,limon2018nonlinear}, tailored to this case.

\begin{lemma} \label{lem:WTC}
{Consider that Assumptions \ref{HipCont}-\ref{as:Morse} and \ref{as:TC1} hold} and let $V^*_{\Or,tr}(\cdot)$ be the optimal cost of the optimization problem $P_{\Or,tr}(\cdot)$ given by~\eqref{eq:POtr}. Then, there exists certain positive constants $d$ and $ \alpha$ such that for any $z$ satisfying 
$$V^*_{\Or,tr}(z) \leq N_p d + \alpha,$$
the terminal constraint \eqref{eq:PO8} in $P_{\Or,tr}(z)$ is not active.
\end{lemma}
\begin{proof}
Let define the compact region $ \Omega_f =\{ z: V_{\Or,fr}(z) \leq \alpha\}$ and let $\alpha$ be a positive constant such that $\Omega_f$ is contained in the interior of $\setZ_f$. Notice that this constant exists since the constraints are not active at the equilibrium point,~$z=z_s$.

Then, for all $z \in \Omega_f$ we have that 
\begin{equation}
    V_{\Or,fr}(\hat F(z,\kappa_f(z))) - V_{\Or,fr}(z) \leq - L_r(z,\kappa_f(z)).
\end{equation}
From this inequality we have that for any $M\geq 1$
\begin{equation} \label{eq:IneqCT}
   \sum_{j=0}^{M-1} L_r(z(j),\kappa_f(z(j))) +  V_{\Or,fr}(z(M)) \leq V_{\Or,fr}(z),
\end{equation}
where $z(j)$ is the trajectory of the system beginning at $z$, and controlled by the terminal control law $u=\kappa_f(z)$. Besides, from \citep[Proposition 2.18]{RawlingsLIB17} we can infer that $V^*_{\Or,tr}(z) \leq V_{\Or,fr} (z)$, for all~$z \in \Omega_f$.

Define the positive constant $d>0$ such that $L_r(z,u) > d$ for all $z \not \in \Omega_f$, and for all $u \in \setU$. This constant exists since $L_r(\cdot,\cdot)$ is a positive definite function of $(z-z_s)$, and $\Omega_f$ is a neighborhood of $z_s$.

Let $\vu^*$ and $\vz^*$ be the optimal sequences of inputs and states of $P_{\Or,tr}(z)$. We will prove that if the terminal state $z^*(N_p)$ is not contained in $\Omega_f$, then the whole predicted trajectory lies out of $\Omega_f$.

This will be proved by contradiction. Assume that $z^*(N_p) \not  \in \Omega_f$, but that there exists certain instant $i$ in which $z^*(i) \in \Omega_f$. 


Define the trajectories $\bar \vu$ and $\bar \vz$ as follows: for $j \in \mathbb I_{0}^{i-1}$,  $\bar u(j)=u^*(j)$  and  $\bar z(j+1) = \hat F(\bar z(j),u^*(j))$, with $z(0)=z$ (which satisfies that $\bar z(j)=z^*(j)$); and  for $j\geq i$, $\bar u(j)=\kappa_f(\bar z(j))$ and $\bar z(j+1) = \hat F(\bar z(j),\kappa_f(\bar z(j)))$.

Define $\bar V_{\Or,tr}(z)$ as the cost associated to the input sequence $\bar \vu$. Then, since $z^*(i) \in \Omega_f$ and taking into account inequality \eqref{eq:IneqCT}, we have that
\begin{eqnarray*}
\hspace{-0.7cm}\bar V_{\Or,tr}(z) &=& \sum_{j=0}^{i-1} L_r(z^*(j),u^*(j))) 
\\&&+ \sum_{j=i}^{N_p} L_r(\bar z(j),\kappa_f(\bar z(j))) \\
&&+  V_{\Or,fr}(\bar z(N_p)) \\
&\leq & \sum_{j=0}^{i-1} L_r(z^*(j),u^*(j))) +  V_{\Or,fr}( z^*(i)).
\end{eqnarray*}
On the other hand,  taking into account that $z^*(N_p) \not \in \Omega_f$
\begin{eqnarray*}
 V^*_{\Or,tr}(z) &=& \sum_{j=0}^{i-1} L_r(z^*(j),u^*(j)))\\
 &&+ \sum_{j=i}^{N_p} L_r( z^*(j),u^*(j))\\
 &&+  V_{\Or,fr}(z^*(N_p)) \\
&> & \sum_{j=0}^{i-1} L_r(z^*(j),u^*(j))) +  \alpha.
\end{eqnarray*}
Moreover, by optimality we have that $\bar V_{\Or,tr}(z) \geq V^*_{\Or,tr}(z)$, and then
\begin{eqnarray*}
\sum_{j=0}^{i-1} L_r(z^*(j),u^*(j))) +  V_{\Or,fr}( z^*(i)) \geq \bar V_{\Or,tr}(z)\\
\geq V^*_{\Or,tr}(z) > \sum_{j=0}^{i-1} L_r(z^*(j),u^*(j))) +  \alpha.
\end{eqnarray*}

From this, we can state that $V_{\Or,fr}( z^*(i))> \alpha$, and hence $z^*(i) \not \in \Omega_f$, which is a contradiction.

Therefore, we have proved that if $z^*(N_p)  \not \in \Omega_f$, then $z^*(j) \not \in \Omega_f$ for all~\mbox{$j\in \mathbb I_{0}^{N_p}$}, which implies that $L_r(z^*(j),u^*(j))>d$. Then, the optimal cost function is such that
$$V^*_{\Or,tr}(z) > N_p d + \alpha.$$

Consequently, it is proven that for any $z$ such that $V^*_{\Or,tr}(z) \leq  N_p d + \alpha$, then~$z^*(N_p) \in \Omega_f \subset \setZ_f$, and therefore the terminal constraint is satisfied.
\end{proof}

{At this point, we have demonstrated sufficient conditions to derive the robust stability of the proposed oracle-based predictive controller in presence of inexact oracles: (i) the control law ensures asymptotic stability in nominal conditions, i.e. assuming that the oracle provides exact predictions; (ii) the oracle provides predictions with bounded errors for any input applied on the system and (iii) the terminal constraint is not active which, as it will be proven next, ensures uniform continuity of the optimal cost function. This result is rigorously stated in the following theorem.}

\begin{theorem} \label{ThmISS}
   {Consider that Assumptions \ref{HipCont}-\ref{as:Morse} and \ref{as:TC1} hold,} then the system \eqref{EqMdl}  controlled by the control law derived from $P_{\Or,t}(z)$ is input-to-state stable w.r.t. the model mismatch signal $w$ in a neighborhood of the equilibrium point $z_s$, $\Omega_r$, for a sufficiently small model mismatch. Besides, the region $\Omega_r$ is enlarged if the prediction horizon increases, i.e.,
    \begin{equation}
        \Omega_r(N_p) \subseteq \Omega_r(N_p+1).
    \end{equation}
\end{theorem}

\begin{proof}

From Lemma~\ref{lem:WTC}, we have that there exists some positive constants $d$ and $\alpha$ such that for all $z$ satisfying
\begin{equation}
    V^*_{\Or,tr}(z) \leq N_p d + \alpha, \label{eq:RegVN}
\end{equation}
the terminal constraint in $P_{\Or,t}(z)$ is not active.

Define $\Omega_r(N_p)=\{z: W(z) \leq N_p d + \alpha\}$, which is a level set for the Lyapunov function, and hence it is a positive invariant set of the controlled system. Besides, since $W(z) \geq \beta_1(|z-z_s|)$, $\Omega_r$ is compact.  
From this definition we have that for all $z(k)\in \Omega_r$, $V^*_{\Or,tr}(z(k)) \leq W(z(k))  \leq N_p d + \alpha$, and thus the terminal constraint is not active. 



Since $\Omega_r$ is compact and $u$ is bounded, the value of $\Or(z,u)$ is bounded for all $(z,u) \in \Omega_r \times \setU$. Therefore, as the function $\Or(z,u)$ is continuous, it is also uniformly continuous in $\Omega_r \times \setU$, and then the model function \eqref{EqOracle} is uniformly continuous. 
Hence, in virtue of \cite[Theorem 4]{LimonNMPC08}, the system controlled by $\kappa_{\Or,t}(z)$ is ISS with respect to the model mismatch signal $w$ in $\Omega_r$, if the model mismatch is sufficiently small.

Finally, from the definition of $\Omega_r(N_p)$ we have that $\Omega_r(N_p) \subseteq \Omega_r(Np+1)$ since for any feasible $z$, the optimal cost function decreases with $N_p$.
\end{proof}

\begin{remark}
    In the previous section it was proven that for a sufficiently large prediction horizon, the control law without terminal constraint derived from~\eqref{eq:POtc} is nominally stabilizing. Since this optimization problem does not have constraints on the states, following the arguments of Theorem \ref{ThmISS}, the closed-loop system is ISS w.r.t. to the model mismatch.
\end{remark}
\begin{remark}[Online learning] 
Input-to-state stability of the oracle-based EMPC can be extended to the case that the oracle is updated online with fresh data collected during the closed-loop operation of the system. Notice that occasionally, the update of the data set might not lead to an enhancement of the estimation error. Under the assumption that the effect of the update policy on the predictions is bounded by a sufficiently small value, the possible increment on the predicted cost function is also bounded and sufficiently small, and then, the optimal cost function serves also as an ISS-Lyapunov function.\\
This condition can be relaxed if the update policy filters those data points collected that worsen the expected cost function, as proposed in \citep{ManzanoIJRNC20}.
\end{remark}

\section{Case study}\label{sec:casestudy}

This section presents a simulated case study in which the system taken into consideration is the continuously stirred tank reactor (CSTR) introduced in~\cite{FaulwasserIECR19}.

\subsection{The reactor}

 Two consecutive reactions take place: $A\to B$ at a rate~$k_1=\SI{1}{\per\minute}$, and~$B\to C$ at a rate~$k_2=\SI{0.05}{\per\minute}$.
The molar concentrations of~$A$ and~$B$ are denoted~$c_A$ and~$c_B$, and measured in~\si{\kilo\mol\per\cubic\meter}. They evolve according to the following set of ordinary differential equations:
\begin{subequations}
    \begin{eqnarray}
    \hspace{-0.7cm}\frac{dc_A(t)}{dt}&=&\frac{u(t)}{V}(c_{A0}-c_A(t))-k_1c_A(t),\\
    \hspace{-0.7cm}\frac{dc_B(t)}{dt}&=&\frac{u(t)}{V}(-c_A(t))+k_1c_A(t)-k_2c_B(t),\nn\\
    \end{eqnarray}
\end{subequations}
\noindent where~$u$~(\si{\cubic\meter\per\minute}) is the manipulable input: the feed flow rate (note that the same flow exits the tank, so the volume remains constant,~$V=\SI{1}{\cubic\meter}$), and~$c_{A0}=\SI{1}{\kilo\mol\per\cubic\meter}$ is concentration of~$A$ in the input flow. The system is discretized with sampling time~$\tau_s=\SI{0.25}{\minute}$.

The economic cost function accounts for the use of reactant and the benefits of the product:
\begin{equation}
    \ell(x(k),u(k))=u(k)(c_{A0}-\alpha c_B(k)),
\end{equation}
where the weighting term is set to~$\alpha=4$ (in the appropriate dimensions). Existing noise in the measurements follows a normal distribution of zero mean and standard deviation of~$2\%$ of the measurement.
The system is subject to input constraints, such that~$0< u\leq \SI{2}{\cubic\meter\per\minute}$. The set of equilibrium points is represented in Figure~\ref{fig:k_est}. Its minimum is obtained for~$u_s=\SI{1.043}{\cubic\meter\per\minute}$, resulting in~$\ell_s=\SI{-0.957}{\$\per\minute}$.

\begin{figure}[t]
    \centering
    \includegraphics[width=1\columnwidth]{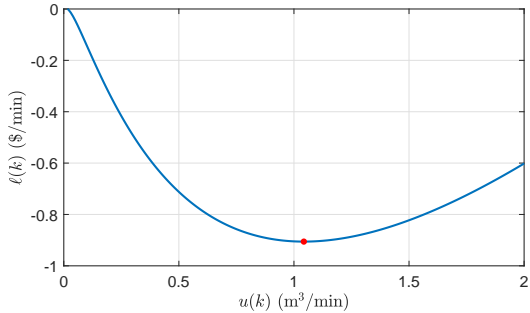}
    \caption{Set of equilibrium costs.}
    \label{fig:k_est}
\end{figure}

\subsection{Ideal estimation}

To replicate the nominal conditions studied in the paper, we seek to obtain an oracle that provides an exact estimation, in order to compare the control actions that the standard EMPC and the oracle-based EMPC yield, and to illustrate that they are the same, as proven in Theorem~\ref{th:eq}.

To this end, consistent states (as introduced in Definition~\ref{def:consistent}) must be obtained, and applied to the mentioned controllers. To do so we proceed as follows. An initial state is chosen randomly, following an uniform distribution, within the workspace (both concentrations among~$0$ and~\SI{1}{\kilo\mol\per\cubic\meter}). A sequence of control actions of length~$N_p=\max(n_a,n_b)$, with~$n_a=3,\,n_b=2$ is chosen, each element drawn randomly from~$\setU$ using an uniform distribution. This sequence is applied to the real system, which yields the initial state~$x_0$, consistent with the initial regressor~$z_0$, given the measured costs. This setup is represented in Figure~\ref{fig:consistent}.

\begin{figure}[t]
    \centering
    \includegraphics[width=1\columnwidth]{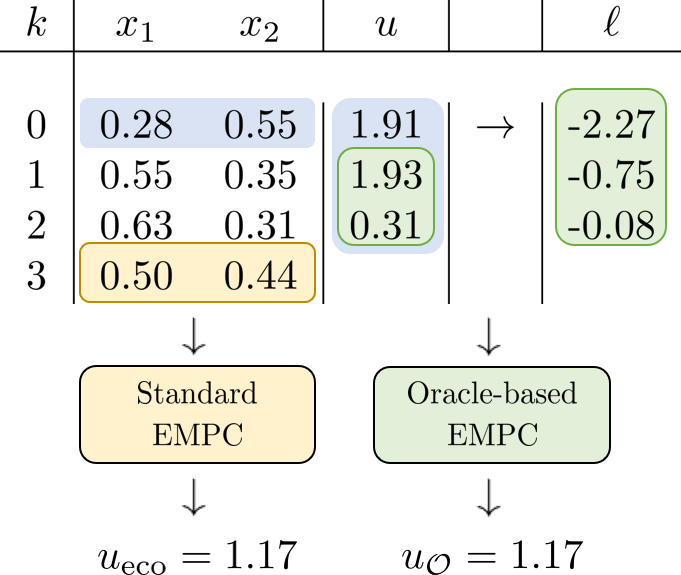}
    \caption{One experiment of the procedure to compare standard and oracle-based economic predictive controllers. Blue values are initial conditions obtained randomly, which yield the consistent states~$x_0$ (in yellow, fed to the standard EMPC) and~$z_0$ (in green, fed to the oracle-based EMPC). Building an exact oracle and applying~$P^e_N(x_0)$ and~$P^e_\Or(z_0)$ yield the same control action.}
    \label{fig:consistent}
\end{figure}

Next, the standard EMPC with terminal equality constraint given by the optimization problem~\eqref{eq:PNe} is applied from~$x_0$, with prediction horizon~$N=5$ and run for a single time step. Every time the solver in the optimization problem (MATLAB\textsuperscript{\textregistered}'s {\tt fmincon} function) obtains the predicted trajectory and cost for a candidate sequence~$\mathbf u$, the resulting regressors are stored in a data set~$\setD$.

Once a data set is obtained, the kinky inference (KI) class of learning rules is applied to build the oracle. Kinky inference~\cite{calliess2014_thesis} is a class of non-parametric learning methods based on Lipschitz interpolation, which have proven able to conform a valid model for nonlinear learning-based predictive controllers~\cite{manzano2020robust}. Further details of the functioning of KI can be found in~\cite{manzano2019output}. Given its non-parametric nature, it is guaranteed that in the absence of noise the prediction over a given query~$q$ already contained in the data set is its corresponding real output~$f(q)$, which is key for the exact estimation sought in this paper. Besides, there is no tuning process, apart from identifying the Lipschitz constant of the ground truth function, which is overweighted to~$L=100$ in order to ensure sample consistency.

Once the oracle~\eqref{EqOracle} is built, the oracle-based MPC problem~\eqref{eq:PO} is applied from~$z_0$, again with~$N=5$ and for a single simulation step. It is important to remark that despite the way the data set was obtained, the only information given to the oracle (and to the optimization problem) is measurements of past costs and input signals, as we assumed in this paper. Hence, note that neither the states nor the control law of the standard EMPC have been accessible to the oracle, as represented in Figure~\ref{fig:consistent}.

This setup is repeated a hundred times, with different initial random states and input sequences. The resulting optimal control actions are the same in all cases, as it was stated in Theorem~\ref{th:eq}, $\kappa_\Or^e(z(k))=\kappa_{\textup{eco}}^e(x(k))$ for all~$k$, proving the statements of this paper.

\subsection{Closed-loop example}

The previous example illustrates that under perfect estimation, the control decision of the proposed oracle-based EMPC is the same as the one that would be obtained with a standard EMPC. However, the conditions to obtain the data set are artificial and not practical. Next, we practically demonstrate that an oracle derived from realistic input-output data sets can be obtained, and that  the proposed oracle-based economic MPC stabilizes the real plant in presence of estimation errors.

To this end, a chirp signal is designed with initial and final frequencies of~\num{1d-6} and~\SI{0.3}{\per\minute}, respectively; and amplitude and center of~\SI{1}{\cubic\meter\per\minute}, of~\num{40d3} samples of length. This signal is applied as control input to the plant, yielding the data set of costs-inputs represented in Figure~\ref{fig:data_chirp}. The regression state is built with~$n_a=3,\,n_b=2$.

\begin{figure}[t]
    \centering
    \includegraphics[width=1\columnwidth]{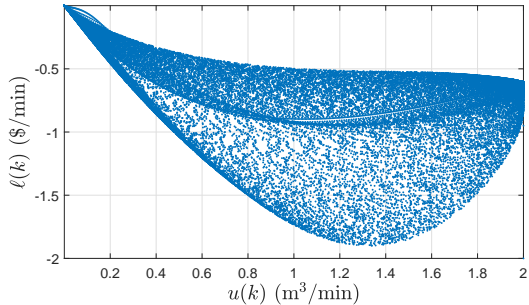}
    \caption{Training data set, obtained with a chirp input signal.}
    \label{fig:data_chirp}
\end{figure}

A kinky inference predictor is built with these data. In order to asses its performance, another set of experiments is carried out, applying a pseudorandom signal to the plant that consists of a sequence of steps of value~$0\leq u\leq 2$ (\si{\cubic\meter\per\minute}) and length among~\SI{30}{\second} and~\SI{5}{\minute}. The validation test represented in Figure~\ref{fig:val} shows the performance and the prediction error.

Then, the oracle-based MPC~\eqref{eq:PO} is applied, from one hundred initial points consisting of the regressor obtained applying a random input sequence to the plant on a random state, and measuring costs. The results are shown in Figure~\ref{fig:cl}. Besides, in order to compare the performance to the ideal case, the standard EMPC~\eqref{eq:PNe} is applied to the same 100 initial states. Note that since the estimation is not perfect, the response differ, although the closed-loop performance is very similar. Such performance, represented in Figure~\ref{fig:cl_hist}, is assessed by the index:
\begin{equation}\label{eq:phi}
    \Phi=\sum_{i=1}^{t_\mathrm{sim}} \ell(i).
\end{equation}

\begin{figure}[t]
	\centering
	\subfloat[Validation experiment]{\includegraphics[width=1\columnwidth]{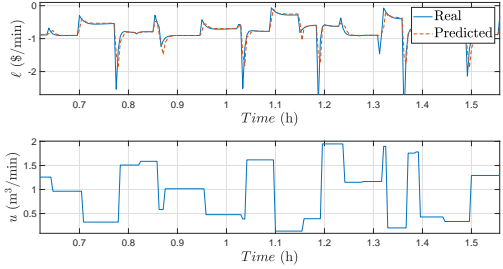}}\hfil
	\subfloat[Histogram of the prediction error]{\includegraphics[width=1\columnwidth]{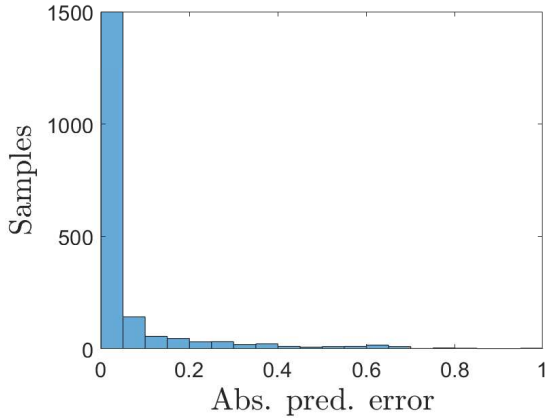}}\hfil
	\caption{Validation of the predictor} \label{fig:val}
\end{figure}

\begin{figure}[t]
	\centering
	\subfloat[Closed-loop]{\includegraphics[width=1\columnwidth]{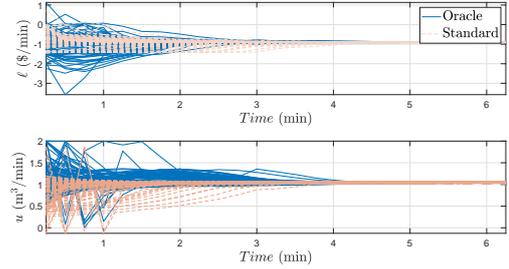}}\hfil
	\subfloat[Performance index]{\includegraphics[width=1\columnwidth]{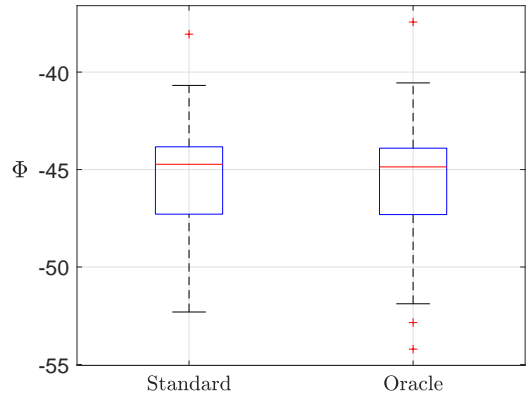}\label{fig:cl_hist}}\hfil
	\caption{Closed-loop performance of the oracle-based EMPC (blue solid lines) and the standard ideal state-feedback EMPC (orange dashed lines), for 100 initial random states.} \label{fig:cl}
\end{figure}

\subsection{Online learning}

Using the same setup as in the previous case, the system is disturbed with an impulse, periodically every~\SI{10}{\minute}, returning to the same initial state. At that same time step, the data set~$\setD$ is updated, including the set of regressors and costs measured during the transient. The kinky inference technique employed to learn the oracle can be updated recursively without a cumbersome tuning. The recalculation of the Lipschitz constant can be done recursively in~$\mathcal O(n_\setD)$~\cite{calliess2020lazily}, being~$n_\setD=40$ the number of new data points.

The performance index~\eqref{eq:phi} is computed for each one of the iterations, that is, during the interval of \SI{10}{\minute} after the impulse when the system is regulated to the optimal equilibrium point rejecting its effect. The value of this index for every iteration is shown in Figure~\ref{fig:online}. Note that, in general, the inclusion of past measurements in an online learning fashion enhances the performance, yielding better closed-loop responses, closer to the ideal state-feedback ones.

\begin{figure}[t]
    \centering
    \includegraphics[width=1\columnwidth]{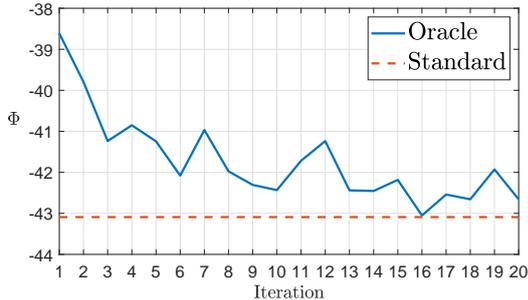}
    \caption{Performance index each iteration, i.e. during the regulation between disturbing the system, including past measurements in the oracle online.}
    \label{fig:online}
\end{figure}

\section{Conclusion}
The objective of this paper was to prove that an economic predictive controller can be designed for a system whose model is unknown, under the assumption that only the economic cost of the plant is observed. This set-up can be encountered in many situations in which the client is not willing to provide further details of the intern operation of his plant, or when the system is too complex to be modeled. 

The evolution of the economic cost is predicted using an oracle, i.e., a data-based prediction model of the cost that is learnt from a set of historical input-cost measurements of past trajectories. The structure of the oracle is a NARX model, and the prediction is done with a chosen machine learning technique.

First, it is demonstrated that under mild assumptions on the economic cost function, an oracle can be obtained for generic nonlinear plants. Then, it is proven that in nominal conditions, an oracle-based economic MPC with terminal equality constraint provides the same control law as a state-space prediction model. Besides, stabilizing design of oracle-based economic MPC using terminal ingredients is presented. Finally, it has been formally proven that the resulting control law can stabilize the plant in a real setting, where the oracle may exhibit estimation error.

These results demonstrate that under mild assumptions on the model of the plant and on the economic cost function, an oracle of the economic cost suffices to design economic predictive controllers with the same good properties of the one based on state-space.

These properties have been demonstrated in an realistic illustrative case study where a CSTR is regulated with an economic predictive control  based on first principle model and on an appropriate oracle derived using kinky inference from input-output data of the plant.


\section*{Acknowledgments}
	The authors would like to thank the funds received by the MINECO Spain and Feder Funds under contract DPI2016-76493-C3-1-R and Ministerio de Ciencia e Innovación of Spain under project PID2019-106212RB-C41.

\appendix
\section{Existence of a NARX model}
Consider a discrete-time system described by
\begin{eqnarray}
x(k+1)&=&f(x(k),u(k)) \label{ap:sys1}\\
y(k)&=& h(x(k)), \label{ap:sys2}
\end{eqnarray}
where $x \in \mathcal X \subset \R^n$, $u \in \mathcal U \subset \R^m$ and $y \in \mathcal Y \subset \R^p$, being $\mathcal X$, $\mathcal U$ and $\mathcal Y$ compact sets. The functions $f(\cdot,\cdot)$ and $h(\cdot)$ that define the model are considered to be unknown, and the state not measurable. The only measures available for this system are the input $u(k)$ and the output $y(k)$.

In this section we summarize the results presented by Levin and Narendra \footnote{In order to present the Levin and Narendra's results in a form closer to the notation used in this paper, we present them taking the sequence of inputs and outputs as a collection of past measurements, contrary to the original paper where these were defined as sequence of future values.} in \cite{LevinCAP97}, on the existence of an input-output description of the system by means of a nonlinear autoregresive with exogenous signal (NARX) model, of the form
\begin{equation}
y(k+1)=\mathcal F(\vy_M(k),\vu_M(k)), \label{ap:NARX}
\end{equation}
where
\begin{align} 
\vy_M(k)=(y(k),y(k-1),\cdots,y(k-M)) \subset \mathcal Y^{M+1}\\
\vu_M(k)=(u(k),u(k-1),\cdots,u(k-M))\subset \mathcal U^{M+1}.
\end{align}

The existence of the NARX model is closely related with the notion of global observability of the system. Levin and Narendra define the observability map as\footnote{Formally, the sequence $\vy_M(k)$ depends on $\vu_{M-1}(k-1)$, but this is extended by Levin and Narendra w.l.o.g. to $\vu_M(k)$, for the sake of simplicity of the presented results.}
\begin{equation}
\vy_M(k)=\mathcal H_M(x(k-M),\vu_{M}(k))
\end{equation}

Then, a system is said to be observable if for any two distinct states $x_1$ and $x_2$, there exists an input sequence  $\vu_{M}$ such that the corresponding output sequences are distinct, i.e.  $\mathcal H_M(x_1,\vu_{M})\neq \mathcal H_M(x_2,\vu_{M})$. This can be read as the map $\mathcal H_M(x,\vu_{M})$ being injective in $x$.

If this property holds for any input sequence $\vu_{M}$, then the system is said to be strongly observable, and if it holds for almost every input sequence $\vu_{M}$, then the system is called generically observable. In this case, the set of input sequences for which this property does not hold $\setU_{no}$ is of measure zero.

\begin{assumption} \label{ap:Assumpt1}
The system given by \eqref{ap:sys1} and \eqref{ap:sys2} satisfies the following conditions:
\begin{enumerate}
\item The model functions $f$ and $h$ are smooth.
\item The system \eqref{ap:sys1} is state invertible, i.e. for a given $u$, $f(x,u)$ defines a diffeomorphism on $x$.
\item The function $h$ is a Morse function, i.e. its critical points are non-degenerate (see Definition~\ref{def:morse} on page~\pageref{def:morse}).
\end{enumerate}
\end{assumption}

The second condition is naturally met by continuous-time sampled systems controlled by a discrete-time measurable control law function and a zero-order holder.

Under this assumption, Levin and Narendra derived Theorem 6 in~\cite{LevinCAP97}, on the existence of an input-output model. This Theorem is rewritten here as follows:

\begin{theorem} \label{ap:ThmNARX}
Let a system be defined by \eqref{ap:sys1} and \eqref{ap:sys2} satisfying Assumption~\ref{ap:Assumpt1}. Then, for $M\geq 2n$ and for any $\epsilon>0$, there exists a set of sequences of inputs $\mathcal \setU_M^\epsilon \subseteq \setU^{M+1}$ of measure $\mu(\mathcal \setU_M^\epsilon)< \epsilon$, such that\\
\begin{enumerate}
\item The system defined by \eqref{ap:sys1} and \eqref{ap:sys2} {is observable for any sequence~{$\vu_M \not \in \setU_M^\epsilon$}}.
\item  There exists  a continuous and bijective map $\Phi$ such that for all $x(k-M) \in \mathcal X$ and $\vu_M(k) \in \mathcal \setU^{M+1} \setminus \mathcal \setU_M^\epsilon$ we have  
\begin{equation}
(x(k-M),\vu_{M}(k))=\Phi(\vy_M(k),\vu_{M}(k)).
\end{equation}
\item There exists a continuous function $\mathcal F(\cdot,\cdot)$  such that for all input sequences $\vu_M(k) \in \mathcal \setU^{M+1} \setminus \mathcal \setU_M^\epsilon$ we have
\begin{eqnarray}
y(k+1)&=&\mathcal F(\vy_M(k),\vu_M(k)).
\end{eqnarray}

\item There exists a continuous function $\mathcal{\tilde F}$ such that
for all input sequences $\vu_M(k) \in \mathcal \setU^{M+1}$ we have
{
$$\| \mathcal F(\vy_M(k),\vu_M(k))-  \mathcal{\tilde F}(\vy_M(k),\vu_M(k))\| < \epsilon.$$
}
\end{enumerate}
\end{theorem}

\begin{proof}
{First of all, notice that from~\cite[Theorems~4]{LevinCAP97}, it is derived that the generic observability property holds for almost all systems satisfying Assumption~\ref{ap:Assumpt1}.}

The first statement is proven in the first line of the proof of \cite[Theorem~6]{LevinCAP97}.  The existence of the map $\Phi$ in the second statement is demonstrated  in~\cite[Theorem 5]{LevinCAP97}. The existence of $\mathcal F$ in the third statement corresponds to statement 1 of~\cite[Theorem 6]{LevinCAP97}. The last statement is derived from the proof of the second part of \cite[Theorem 6]{LevinCAP97}, for the case in which the continuous function~$\mathcal{\tilde F}$ is given by a multilayer feedforward neural network with a sigmoidal function as activation function of each neuron.
\end{proof}


{
The following corollary proves that there exists an input-output system such that its estimation error w.r.t. the real output signal is arbitrarily small for any sequence of inputs, irrespective if they are observable or not.}

{
\begin{corollary}
Under the assumptions of Theorem \ref{ap:ThmNARX}, there exists a continuous function $\mathcal{\tilde F}$ such that for all input sequences $\vu_M(k) \in \mathcal \setU^{M+1}$, $$\| y(k+1) -  \mathcal{\tilde F}(\vy_M(k),\vu_M(k))\| < \Theta(\epsilon),$$
for a $\mathcal K$-function $\Theta(\cdot)$.
\end{corollary}
}
{
\begin{proof}
This can be proven as a consequence of Proposition \ref{PropErrorOracle} (on page~\pageref{PropErrorOracle}), taking~$\mathcal F(Y_M,U_M)$ as a continuous oracle function which provides a null estimation error (i.e. $\mu=0$). Therefore, from this Proposition we infer that 
$$ |y(k+1) - \mathcal F(\vy_M(k),\vu_M(k)|\leq \Theta(\epsilon).$$
\end{proof}
}

{Last, the following corollary extends the formulation of Theorem~\ref{ap:ThmNARX} to feed-through systems, as they are considered in this paper.}

\begin{corollary} \label{ap:CorNARX}
Assume that the output of the system $y(k)$ depends explicitly on $u(k)$, i.e.
\begin{equation}
y(k)=h(x(k),u(k)),
\end{equation}
such that $h$ is Morse in its first argument, i.e., for any critic point $(x_a,u_a)$ of $h$ its Hessian in $x$ is nonsingular. Then, Theorem \ref{ap:ThmNARX} holds considering in the third claim the following NARX model:
\begin{eqnarray}
y(k+1)&=&\mathcal F(\vy_M(k),\vu_{M+1}(k+1)).
\end{eqnarray}
\end{corollary}
\begin{proof}
Notice that in this case, the considered observability map is valid, making most of the results immediate.  In the proof of Lemma 1 in~\cite{LevinCAP97}, equation (30) should be modified accordingly, but the subsequent equation (31) holds true in virtue of the chain rule, and therefore, the proof of the lemma follows. The NARX model is derived taking into account that $y(k+1)$ can be written as a continuous function of the state $x(k-M)$ and the sequence $\vu_{M+1}(k+1)$.
\end{proof}

\bibliography{bibPepe}

\end{document}